\documentclass[twoside,leqno]{article}

\newif\ifsubmission
\submissiontrue
\submissionfalse %

\newif\ifcomment
\commenttrue

\newif\ifcamera
\cameratrue
\camerafalse %

\ifsubmission
\commentfalse
\author{Anonymous Authors}
\else
\author{
  Yotam Kenneth-Mordoch
  \qquad
  Robert Krauthgamer%
  \thanks{ The Harry Weinrebe Professorial Chair of Computer Science.
    Work partially supported by the Israel Science Foundation grant \#1336/23,
    by the Israeli Council for Higher Education (CHE) via the Weizmann Data Science Research Center,
    and by a research grant from the Estate of Harry Schutzman.
  }
  \\ Weizmann Institute of Science
  \\ \texttt{\{yotam.kenneth,robert.krauthgamer\}@weizmann.ac.il}
}

\fi

\usepackage[letterpaper, margin=1in]{geometry}
\ifcamera
\usepackage{siamproceedings}
\usepackage{amsmath}

\else
\usepackage{amsmath}
\usepackage{amsthm}
\usepackage[hypertexnames=false]{hyperref}
\usepackage[capitalise]{cleveref} 

\newtheorem{theorem}{Theorem}[section]

\newtheorem{proposition}[theorem]{Proposition}
\newtheorem{lemma}[theorem]{Lemma}
\newtheorem{corollary}[theorem]{Corollary}

\newtheorem{definition}[theorem]{Definition}

\hypersetup{colorlinks=true,allcolors=blue}

\fi

\usepackage{amsmath}
\usepackage{algorithm}
\usepackage[noend]{algpseudocode}
\usepackage{amssymb}
\usepackage{amsfonts}
\usepackage{import} 
\usepackage{xifthen}
\usepackage{mathtools}
\usepackage{arydshln}%
\usepackage{dsfont}
\usepackage[cal=esstix]{mathalfa}
\usepackage{thm-restate}

\usepackage{comment}
\usepackage[colorinlistoftodos,textsize=tiny,textwidth=2cm,color=red!25!white,obeyFinal]{todonotes}

\hypersetup{colorlinks=true,allcolors=blue}

\def\compactify{\itemsep=0pt \topsep=0pt \partopsep=0pt \parsep=0pt}

\ifdefined\AddToHook
\AddToHook{cmd/appendix/before}{\crefalias{section}{appendix}}
\fi
\newcommand{\mintcut}{\mathrm{cut}}

\newcommand{\C}{C}

\newcommand{\R}{\mathbb{R}}
\DeclareMathOperator*{\E}{\mathbb{E}}
\newcommand{\Exp}[2]{\E_{ #2}\left[ #1 \right]}

\newcommand{\Probability}[1]{\Pr\left[ #1 \right]}
\newcommand{\ProbOn}[2]{\Pr_{#2}\left[ #1 \right]}

\newcommand{\tO}{\tilde{O}}

\DeclareMathOperator{\poly}{poly}
\DeclareMathOperator{\polylog}{polylog}
\newcommand{\indic}[1]{ \mathds{1}_{\{#1\}} } %
\providecommand{\set}[1]{{\{#1\}}}
\newcommand{\eqdef}{\coloneqq}

\newcommand{\Uc}{U(\C)}
\newcommand{\dd}{\mathbf{d}}
\newcommand{\BalCutPrune}{\text{BalCutPrune}}
\newcommand{\HH}{\mathcal{H}}

\newcommand{\F}{F}

\newcommand{\cross}[2]{\crosscut_{#2}(#1)}
\newcommand{\crossbar}[2]{\overline{\crosscut}_{#2}(#1)}

\newcommand{\Star}{P}

\newcommand{\cutedges}[2]{E_{#2}(#1,V\setminus #1)}
\DeclareMathOperator{\crosscut}{cr}

\newcommand{\colnote}[3]{\textcolor{#1}{$\ll$\textsf{#2}$\gg$\marginpar{\tiny\bf #3}}}

\ifcomment
\newcommand{\rnote}[1]{\colnote{purple}{#1--Robi}{RK}}
\newcommand{\ynote}[1]{\colnote{blue}{#1--Yotam}{YK}}
\else
\newcommand{\rnote}[1]{}
\newcommand{\ynote}[1]{}
\fi

\newtheorem{claim}[theorem]{Claim}

\ifcamera
\else

\fi
\title{All-Pairs Minimum Cut using $\tilde{O}(n^{7/4})$ Cut Queries}

\date{}
\begin{document}

\maketitle

\begin{abstract}
We present the first non-trivial algorithm for the all-pairs minimum cut problem in the cut-query model.
Given cut-query access to an unweighted graph $G=(V,E)$ with $n$ vertices,
our randomized algorithm constructs a Gomory-Hu tree of $G$,
and thus solves the all-pairs minimum cut problem,
using $\tilde{O}(n^{7/4})$ cut queries.
\end{abstract}

\section{Introduction}
\label{sec:introduction}
Optimization over graphs, such as finding a minimum $s,t$-cut, is a fundamental topic in computer science due to its wide range of applications from network design to combinatorial optimization and computational complexity.
Accordingly, there exists a vast body of work on solving such problems efficiently in many different models of computation.
One recently suggested model is the \emph{cut-query model} \cite{RSW18},
where algorithms access the input, which is an unweighted graph $G=(V,E)$,
only through cut queries:
the query specifies $S\subseteq V$
and obtains in return the size of the corresponding cut,
i.e. $\mintcut_G(S) \eqdef |E(S,V\setminus S)|$,
where $E(S,T)$ is the set of edges with one endpoint in $S$ and the other in $T$.

There are two main motivations for studying the cut-query model.
The first one is scenarios where it is costly or even impossible to access graph edges directly,
for example in genomic sequencing that models the structure of a genome as a graph.
It is impossible to directly access the graph edges, but one can perform an experiment to compute the number of edges crossing a cut in the graph, and it is imperative to minimize the number of experiments  \cite{GK98,ABKRS02,BGK05,CK08}.%
\footnote{The access model used is additive queries, where given a set $S\subseteq V$ the query returns the number of edges in $E\cap S\times S$. It is easy to see that one can simulate a cut query using $O(1)$ additive queries.}
Another example is distributed systems,
where many low-power devices with limited communication capabilities form a graph.
A cut query can be implemented by a central server sending a message to all devices in a subset $S\subseteq V$ and asking them to report the number of their incident edges crossing the cut, which is more efficient than asking them to report all their neighbors, as in the standard neighbor-query model.

The second motivation is that graph-cut functions are a natural example of submodular functions, and hence the cut-query model is a natural model for studying submodular function minimization (SFM).
A \emph{submodular function} is a set function $f:2^V\to \mathbb{R}$ that satisfies the diminishing returns property
\begin{equation*}
    \forall S \subseteq T\subseteq V,
    \forall v\in V\setminus T,
    \qquad f(S\cup \{v\}) - f(S) \ge f(T\cup \{v\}) - f(T)
    .
\end{equation*}
Denoting $n=|V|$,
the minimum of a submodular function $f$ can be found using $\tO(n^2)$ value queries to $f$ \cite{LSW15,MN21,Jiang23}, 
where we use the usual notation $\tO(\cdot)$ to hide polylogarithmic factors in $n$.
Meanwhile, the known lower bound for SFM is only $\Omega(n)$ queries \cite{GPRW20,LLSZ21}, and $\Omega(n\log n)$ queries for deterministic algorithms \cite{CGJS22}.
The special case of graph cuts therefore poses an intriguing case study to either find stronger lower bounds or inspire better algorithms for SFM.
In particular, while global minimum cut can be found using $O(n)$ queries,
known algorithms for minimum $s,t$-cut use $O(n^{5/3})$ queries \cite{RSW18,MN20,AEGLMN22,ASW25},
leaving a substantial gap to the $\Omega(n)$ query lower bound
(notice that it is a submodular minimization problem on $n-2$ variables).

We study the all-pairs minimum cut problem, whose objective is to find a minimum $s,t$-cut for every pair of vertices $s,t\in V$ in an unweighted graph $G=(V,E)$.
Recent advances reduce this problem to solving $\polylog n$ instances of minimum $s,t$-cut on contracted subgraphs of $G$, and the reduction itself only takes time that is linear in $|E|$ \cite{AKT20,AKT21,JPS21,AKT22a,AKT22b,AKLPST22,AKLP+25}.%
\footnote{Since these graphs are contracted subgraphs of $G$, they may have parallel edges, preventing the direct application of existing minimum $s,t$-cut algorithms in the cut-query model.
For further discussion, see \Cref{sec:technical-overview}.
}
It seems likely that these techniques extend to the cut-query model,
yielding a reduction to $\polylog n$ instances of minimum $s,t$-cut
also in the cut-query model.
Since minimum $s,t$-cut is a submodular minimization problem (on $n-2$ variables),
such a reduction would imply that the query complexity of the all-pairs minimum cut problem is at most $\polylog n$ times that of SFM. 
In that case,
proving a truly super-linear (exceeding linear by more than polylogarithmic factors)
lower bound for all-pairs minimum cut 
would surpass the currently known lower bounds for SFM.
Conversely, our algorithmic results rule out establishing an $\tilde{\Omega}(n^2)$ lower bound for SFM via the all-pairs minimum cut problem in unweighted graphs.

Note that it is possible to recover the entire graph $G$ using $O(n^2)$ cut queries \cite{CK08,BM11,RSW18},
after which the algorithm can solve any minimization problem without additional cut queries.
Hence, our focus is on subquadratic query complexity.

Our main result is the first non-trivial algorithm for the all-pairs minimum cut problem in the cut-query model, improving upon the naive approach that recovers the entire graph using $O(n^{2})$ cut queries.
Our algorithm is based on efficiently constructing a Gomory-Hu tree \cite{GH61}, a well-known data structure that compactly represents the minimum $s,t$-cut for every pair of vertices $s,t\in V$ in a graph $G=(V,E)$.
\begin{theorem}
    \label{theorem:gomory-hu-tree}
    There exists a randomized algorithm for constructing a Gomory-Hu tree of an unweighted graph $G=(V,E)$ with $n$ vertices,
    that uses $\tO(n^{7/4})$ cut queries and succeeds with probability $1-1/\poly(n)$.
\end{theorem}

\subsection{Technical Overview}
\label{sec:technical-overview}
We now provide an overview of our algorithms and their technical components,
and along the way explain our two main technical contributions.
The first one is introducing a weaker version of a problem called isolating-cuts.
The key point is that this relaxation strikes a good balance --
it suffices for our intended application of constructing a Gomory-Hu tree,
yet can be solved efficiently using cut queries.
The second contribution is a new randomized contraction procedure,
termed $(\tau,\F)$-star contraction.
This procedure solves minimum $s,t$-cut,
and is actually an adaptation of previous procedures
that solve the easier task of global minimum cut.

Our algorithm for the all-pairs minimum cut problem constructs
for $G$ a \emph{Gomory-Hu tree} \cite{GH61},
which is a weighted tree on the same vertex set $V$ such that
the minimum $s,t$-cut in the tree is also the minimum $s,t$-cut in $G$.
The classical algorithm of Gomory and Hu~\cite{GH61}
is based on iteratively refining the graph into a tree structure,
using $n-1$ calls to a minimum $s,t$-cut procedure. 
Recent work \cite{AKT21,AKT22a,AKT22b,JPS21}
has managed to accelerate this refinement process
using calls to a procedure that solves single-source minimum cut,
and these offer more choices for which $s,t$-cuts to refine by. 
The refinement process then works effectively like a logarithmic-depth recursion,
where the input graph is partitioned recursively,
and the size of each part decreases by a constant factor each time.

Our algorithm follows this general approach,
and one technical difference is that given a partition of the vertices,
instead of solving single-source minimum cut in each part separately,
we solve the problem for all the parts in parallel.
This difference is crucial in the cut-query model,
essentially because the bottleneck is information theoretic,
and our parallel version aggregates information from the different parts.
The following lemma presents the algorithm that constructs a Gomory-Hu tree
using calls to single-source minimum cut;
its proof is similar to the one in \cite{AKT22a}
and appears in \Cref{sec:gomory-hu-algorithm}.

\begin{restatable}{lemma}{ghtreereduction}
\label{lemma:gomory-hu-tree-reduction}
Suppose that given cut-query access to an unweighted graph $G$ on $n$ vertices,
a set of perturbation edges that guarantee unique minimum $s,t$-cuts,
and a partition of the vertex set $V_1,\ldots,V_k$ with pivots $\{p_i\in V_i\}_{i\in[k]}$,
one can return for every $i\in [k]$ and $v\in V_i$ the minimum $p_i,v$-cut $S_v$
using $q(n)$ cut queries and with success probability $1-\rho_F(n)$.

Then, one can compute a Gomory-Hu tree of an input graph on $n$ vertices using $\tO(q(n) + n^{1.5})$ cut queries.
The algorithm is randomized and succeeds with probability $1 - O(\rho_F(n)\cdot\log^2n) - 1/\poly(n)$.
\end{restatable}

The main focus of our work is thus to solve the single-source minimum cut problem in the cut-query model.
We present an efficient algorithm for this problem in the following lemma,
whose proof appears in \Cref{sec:single-source-cuts}. 

\begin{restatable}[Single-Source Minimum Cut]{lemma}{singleSourceMinCut}
\label{lemma:single-source-min-cut}
Given cut-query access to an unweighted graph $G$ on $n$ vertices,
a set of perturbation edges $\tilde{E}$ that guarantee unique minimum $s,t$-cuts,
and a partition of the vertex set $V_1,\ldots,V_k$ with pivots $\{p_i\in V_i\}_{i\in[k]}$,
one can compute for every $i\in [k]$ and $v\in V_i$ the minimum $p_i,v$-cut $S_v$ using $q(n)=\tO(n^{1.75})$ cut queries.
The algorithm is randomized and succeeds with probability $1-1/\poly(n)$.
\end{restatable}

\Cref{theorem:gomory-hu-tree} follows immediately from \Cref{lemma:gomory-hu-tree-reduction,lemma:single-source-min-cut}.
The main technical tool for solving the single-source minimum cut problem
is procedure for (a relaxed version of) the isolating-cuts problem,
which we define next.

\begin{definition}%
In the \emph{isolating-cuts problem},
the input is a graph $G = (V, E)$ and a set $R \subseteq V$, where $|R| \ge 2$,
and the goal is to output for each $v \in R$ some $S_v\subset V$ that is a $(v, R \setminus v)$-minimum cut.
\end{definition}

It was shown in~\cite{LP20,AKT21} how to reduce the isolating-cuts problem
to $O(\log n)$ instances of minimum $s,t$-cut in contracted multigraphs.
The best algorithm for computing minimum $s,t$-cuts
in the cut-query model uses $\tO(n^{4/3}\cdot c^{1/3})$ cut queries,
where $c$ is the value of the cut \cite{RSW18,ASW25},%
\footnote{This bound is not written there explicitly. In fact, the bound stated in \cite{ASW25} is $\tO(n^{5/3}\cdot W)$, where $W$ is the maximum edge weight in $G$, which is only weaker because the value of a minimum $s,t$-cut is at most $nW$.}
but unfortunately in a contracted multigraph $c$ can be as large as $\Omega(n^2)$.
For example, take a random graph $G(n,1/2)$ and contract two random vertex subsets of cardinality $n/4$ into two vertices $s,t$;
the resulting multigraph is likely to have $\Theta(n)$ vertices
and minimum $s,t$-cut of value $\Omega(n^2)$.
Therefore, existing algorithms do not solve such minimum $s,t$-cut instances efficiently.

Our first technical contribution is to introduce a relaxed variant of the isolating-cuts problem, termed \emph{weak isolating cuts},
that suffices for solving the single-source minimum cut problem
yet can be solved efficiently in the cut-query model.%
\footnote{Another relaxation of isolating cuts, aimed at solving global minimum cut, was introduced in \cite{ASW25}.}
In this relaxation, the goal is to compute a minimum $(v,R\setminus v)$-cut
only if it is also a minimum $v,u$-cut for some $u\in R$,
and otherwise any $(v,R\setminus v)$-cut can be returned.
We observe that the existing algorithms for single-source minimum cut work even 
when calls to isolating cuts are replaced with calls to our relaxed version.
Our main technical result is thus to solve the relaxed problem directly,
i.e., without going through minimum $s,t$-cut in contracted multigraphs,
stated as follows. 

\begin{lemma}[Weak Isolating Cuts]
\label{lemma:isolating-cuts}
Given cut-query access to an unweighted graph $G=(V,E)$ on $n$ vertices,
and also a set of vertices $R\subseteq V$ with maximum degree $d=\max_{v\in R} d_G(v)$ and a set of perturbation edges $\tilde{E}$ that guarantee unique minimum $s,t$-cuts,
one can solve the weak isolating-cuts problem on $R$ using $\tO(\min\{nd, n^{7/4}\})$ cut queries.
The algorithm is randomized and  succeeds with probability $1-1/\poly(n)$.
\end{lemma}

The proof of this lemma, which appears in \Cref{sec:isolating-cuts},
is based on splitting the minimum $(v,R\setminus v)$-cuts
into two sets, of friendly cuts and unfriendly cuts.
Given a cut $C\subseteq V$, let $\cross{v}{C}$ be
the number of edges in $E(C,V\setminus C)$ that are incident to $v$,
and let us omit the subscript when clear from context.
Define the \emph{friendliness ratio} of a vertex $v\in V$
(with respect to $C$) as $1-\cross{v}{}/\deg(v)$,
where $\deg(v)$ is the degree of $v$ in $G$.
The cut $C$ is called \emph{$\alpha$-friendly}
if every vertex $v\in V$ has friendliness ratio $1-\cross{v}{}/\deg(v)\ge \alpha$,
and otherwise it is called \emph{$\alpha$-unfriendly}. 
The case of friendly cuts is handled using a friendly cut sparsifier,
which was first introduced in \cite{AKT22b} and is described next.%
\footnote{A slight difference is that we consider a general parameter $\alpha=\alpha(n)$ instead of fixing $\alpha=0.4$.}

Throughout, a \emph{contraction} of $G=(V,E)$ is a graph $H$
obtained from $G$ by contracting some subsets of vertices.
A \emph{super-vertex} is a vertex in $H$ formed by contracting multiple vertices in $G$. 
Observe that contracting two vertices $u,v\in V$
is equivalent to adding an edge of infinite capacity between them. 
This view is useful because now $H$ has the same vertex set as $G$,
and moreover cuts in $H$ are at least as large as those in $G$
(when comparing the same $S\subset V$).
In the cut-query model,
one can simulate access to a contraction $H$ of the input graph $G$
by replacing each super-vertex in $H$ with the vertices in $V$ that form it.

\begin{definition}%
\label{definition:friendly-cut-sparsifier}
An \emph{$(\alpha,w)$-friendly cut sparsifier} of a graph $G$ is a contraction $H$ of $G$
such that for  every $\alpha$-friendly cut $C\subseteq V$ of $G$ with at most $w$ edges we have $\mintcut_H(C)=\mintcut_G(C)$.
\end{definition}

We show through a straightforward extension of \cite{AKT22b}
an algorithm in the cut-query model that recovers (i.e., reports explicitly)
the edges of an $(\alpha,w)$-friendly cut sparsifier of $G$
using $\tO(\alpha^{-1}n\sqrt{w})$ cut queries.
It can be shown that the contraction resulting from the procedure has at most $\tO(\alpha^{-1}n\sqrt{w})$ edges, and hence $H$ is indeed sparser than $G$.
We will eventually set $\alpha=n^{-1/4}$ and $w=n$,
to obtain a sparsifier with $\tO(n^{7/4})$ edges.
The main technical challenge in the algorithm is to find an expander decomposition of $G$, which we achieve using $\tO(n)$ cut queries.
Afterwards, the friendly sparsifier is obtained by applying a contraction defined by the expander decomposition.
We present an algorithm for constructing a friendly cut sparsifier 
in the following lemma, whose proof appears in \Cref{sec:friendly-cut-sparsifiers}. 

\begin{lemma} [Friendly Cut Sparsifier]
\label{lemma:friendly-cut-sparsifier}
Given cut-query access to an unweighted graph $G$ on $n$ vertices
and parameters $\alpha$ and $w$, 
one can recover all the edges of a friendly $(\alpha,w)$-cut sparsifier of $G$ %
using $\tO(\alpha^{-1}n\sqrt{w})$ cut queries.
The algorithm is randomized and succeeds with probability $1-1/\poly(n)$.
\end{lemma}

The primary challenge in proving \Cref{lemma:isolating-cuts} is to handle unfriendly cuts.
Our second technical contribution addresses this by introducing
a new randomized contraction procedure,
called \emph{($\tau,\F$)-star contraction}, 
which aims to contract most of the edges in the graph
while preserving an unfriendly minimum $s,t$-cut for $s,t\in F$
(with constant probability, separately for each pair). 
Previous work used similar procedures for the easier task of
finding a global minimum cut \cite{GNT20,AEGLMN22,KK25},
and their main idea is that the probability of preserving a fixed cut $C\subset V$
is approximately $\prod_{v\in V} \big( 1-\cross{v}{C}/\deg(v) \big)$,
so for instance if only $O(1)$ vertices have small friendliness ratio,
then the cut $C$ is preserved with rather large probability.%
\footnote{For a global minimum cut, the worst-case scenario for that probability can be shown to be when $O(1)$ vertices have friendliness ratio $1/2$ and the rest have friendliness ratio $1$.}
Unfortunately, a minimum $s,t$-cut might be $o(1)$-friendly and preserved with low probability $o(1)$.

Our key insight is that if $C\subseteq V$ is a minimum $s,t$-cut,
then at most one vertex (in fact either $s$ or $t$) %
is $\alpha$-unfriendly, i.e., has friendliness ratio below $\alpha$, 
and thus by excluding it from the contraction
we can leverage the machinery developed previously for global minimum cut. 
A further complication is that the argument why (in effect)
only $O(1)$ vertices have low friendliness ratio
assumes that $|\cutedges{C}{}|/\min_{v\in V}\deg(v)= O(1)$.
This assumption clearly holds for a global minimum cut, %
however the minimum $s,t$-cut value can be much larger than the minimum degree.
To overcome this, we show that
if the contraction excludes vertices of degree $o(|E(C,V\setminus C)|)$,
then again only $O(1)$ vertices have low friendliness ratio.
Unfortunately, this poses another challenge,
because now edges incident to the low-degree vertices are not contracted,
and the number of such edges might be $\Omega(n\cdot |E(C,V\setminus C)|)=\Omega(n^2)$.
Ultimately, we prove that one need only consider the $O(\alpha \cdot |E(C,V\setminus C)|)$ edges of the cut that are not incident to the $\alpha$-unfriendly vertex,
which we indeed exclude from the contraction due to its low friendliness ratio as discussed above.
Therefore, we can indeed set the degree threshold (to exclude contraction)
at $\alpha\cdot |E(C,V\setminus C)|$,
and again argue that only $O(1)$ vertices have low friendliness ratio.

Our contraction procedure is a variant of the $\tau$-star contraction
from \cite{AEGLMN22,KK25}, which is defined as follows.
Let $H=\{v \in V \mid \deg(v)\ge \tau\}$ be the set of high-degree vertices,
and sample a set $\Star\subseteq V$ of so-called \emph{star vertices},
by including each vertex $v\in V$ independently with probability $p=\Theta(\log n/\tau)$.
Then, for each high-degree non-star vertex $u\in H\setminus \Star$,
sample an edge uniformly from $E(u,\Star)$ and contract it (unless $E(u,\Star)=\emptyset$).
It can be shown that the resulting graph $G'$
preserves a fixed global minimum cut with constant probability,
and has $\tO((n/\tau)^2 + n\tau)$ edges with high probability \cite{KK25}.
Our variant introduces an additional feature,  
a prescribed set of ``forbidden'' vertices $F\subseteq V$ that
are always excluded from $H$. 

\begin{definition}%
\label{definition:star-contraction}
\emph{$(\tau,\F)$-star contraction} is a procedure defined for
an unweighted graph $G=(V,E)$ on $n$ vertices, a threshold $\tau\in [n]$,
and a vertex subset $\F\subseteq V$,
and it works as follows:
    \begin{enumerate} \compactify
        \item let $H=\{v \in V\setminus \F \mid \deg(v)\ge \tau\}$ be the set of high-degree vertices not in $\F$;
        \item sample $\Star\subseteq V$ by including each $v \in V$ independently with probability $p=\Theta(\log n/\tau)$; 
        \item for each $u\in H\setminus \Star$, sample an edge uniformly from $E(u,\Star)$ and contract it.
    \end{enumerate}
\end{definition}

To solve the weak isolating-cuts problem,
we apply $(\tau,\F)$-star contraction with parameters $\tau = \Theta(n\alpha)$ and $F=R$
(recall $R$ is given in the input of the isolating-cuts problem).
We now describe at a high level why this contraction preserves with constant probability
a minimum $s,t$-cut for $s,t\in\F$. 
Consider such a cut $\C\subseteq V$ that it is $\alpha$-unfriendly,
and observe that our choice of $\tau$ implies
$|\cutedges{\C}{}| \le\min\set{\deg(s),\deg(t)} \le n \le \tau/(100\alpha)$.
We wish to bound the probability that the cut is not preserved,
i.e., at least one of its edges is contracted. 
Each vertex $v\in H\setminus \Star$ samples one incident edge to be contracted,
and the probability that this edge belongs to the cut $E(\C,V\setminus \C)$
is $\cross{v}{\Star}/d_{\Star}(v)$,
where $\cross{v}{\Star} \coloneqq |E(\C,V\setminus \C) \cap E(v,\Star)|$,
and $d_\Star\eqdef |E(v,\Star)|$.
Therefore,
\begin{equation*}
    \Probability{\text{the cut $\C$ is preserved} \mid \Star}
    \ge \prod_{v\in H} \left(1-\frac{\cross{v}{\Star}}{d_{\Star}(v)}\right)
    ,
\end{equation*}
where by convention $\cross{v}{\Star}/d_{\Star}(v)=0$ whenever $d_{\Star}(v)=0$. 
We bound this expression using the following proposition. %

\begin{proposition}[Proposition 4.1 in \cite{AEGLMN22}]
    \label{proposition:probability-optimization}
    For integer $n\ge 1$ and $0\leq a < 1 \leq b$, define
\begin{align*}
    F(a,b) 
    := 
    \min \qquad
     &\prod_{i\in [n]} (1-x_i)
    \\
     \text{subject to}
    &\sum_{i\in [n]} x_i = b,
    \\
    & 0\le x_i \le a, \qquad \forall i\in [n].
\end{align*}
Then $F(a,b)\ge(1-a)^{\lceil b/a \rceil}$.
\end{proposition}

To apply this proposition we need upper bounds that hold with high constant probability
for $\max_{v\in H} \cross{v}{\Star}/d_{\Star}(v)$
and $\sum_{v\in H} \cross{v}{\Star}/d_{\Star}(v)$.
Recall that at most one vertex (in fact either $s$ or $t$) 
is $\alpha$-unfriendly, 
and assume without loss of generality it is $s$.
We can further show (we skip the argument in this overview) that
$\Exp{\cross{v}{\Star}/d_{\Star}(v) \mid s\not\in P}{\Star} \le  2\crossbar{v}{-s}/\deg(v)$,
where $\crossbar{v}{-s}$ is the number of cut edges incident on $v$ not including the edge $(v,s)$ if it exists.

Before providing those upper bounds, let us argue that
with high probability 
no edge incident to $s$ is contracted.
The probability that $s$ is sampled into $\Star$ is $p= \Theta(\log n/\tau) = o(1)$,
where the last equality holds for a large enough $\tau$,
and we thus assume henceforth that $s\notin \Star$ (otherwise the algorithm fails).
Observe that %
only edges in $E(H,\Star)$ can be contracted,
and because $s$ is neither in $P$ nor in $H$ (recall $s\in F$),
it follows that no edge incident to $s$ is contracted, as claimed. 

Let us also bound $\sum_{v\in H}\crossbar{v}{-s}$, which will be used shortly.
Since $s$ is $\alpha$-unfriendly and $\C$ is a minimum $s,t$-cut,
we have $\cross{s}{} > (1-\alpha)\cdot\deg(s)\ge (1-\alpha)\cdot|\cutedges{\C}{}|$.
Therefore, the number of edges in the cut $\C$ that are not incident to $s$
is at most $\alpha\cdot |\cutedges{\C}{}|$,
and furthermore $\sum_{v\in H}\crossbar{v}{-s} \le 2\alpha\cdot |\cutedges{\C}{}|$
because the sum counts each of these edges at most twice.
Altogether, 
\begin{equation*}
    \Exp{\sum_{v\in H} \frac{\cross{v}{\Star}}{d_{\Star}(v)}}{P}
    \le
    2\sum_{v\in H} \frac{\crossbar{v}{-s}}{\deg(v)}
    \le \frac{4\alpha\cdot |\cutedges{\C}{}|}{\tau}
    \le \frac{1}{25}
    , 
\end{equation*}
where the first inequality uses that $\deg(v)\ge\tau$ for every $v\in H$,
and the last one uses $|\cutedges{\C}{}| \le \tau/(100\alpha)$ established above.
Moreover, the summands above are non-negative, 
hence we can bound their maximum by their sum and obtain
$\Exp{\max_{v\in H}\cross{v}{\Star}/d_{\Star}(v)}{P} \le 1/25$. 
We can thus conclude, by Markov's inequality,
that with constant probability these two quantities are indeed bounded.

We can therefore use \Cref{proposition:probability-optimization}
to conclude that the cut $\C$ is preserved with constant probability.
Using arguments similar to the $\tau$-star contraction,
the number of edges in the contracted graph is bounded, with high probability,
by $O((n/\tau)^2 + n\tau)$ (ignoring edges incident to $\F$),
and for our setting $\tau = \Theta(n\alpha)$ and $\alpha=n^{-1/4}$,
this bound becomes $O(1/\alpha^2 + n^2\alpha) = \tO(n^{7/4})$. 
The following lemma states the guarantees of the $(\tau,\F)$-star contraction
procedure more formally,
and its proof appears in \Cref{sec:sigma-tau-star-contraction}.

\begin{lemma}[Guarantees for $(\tau,\F)$-star contraction]
\label{lemma:tau-sigma-star-contraction}
Let $G=(V,E)$ be an unweighted graph on $n$ vertices, $\F\subseteq V$ a subset of excluded vertices, and $\tau\ge \log^2 n$ a parameter.
Now suppose $s,t\in \F$ are such that $\min\{\deg(s),\deg(t)\} \le \tau/(100\alpha)$
and that $\C\subseteq V$ is a minimum $s,t$-cut that is $\alpha$-unfriendly.
Then a $(\tau,\F)$-star contraction of $G$ produces a contraction $G'$
such that with probability $2/5$ no edge of $\C$ is contracted,
and with probability $1 - \poly(n)$ the following hold:
\begin{enumerate} \compactify
\item The number of edges in $G'$ is at most $O((n/\tau + |\F|)^2 + n\tau)$. 
\item The number of edges in $G'$ that are incident to $\Star$
  is at most $O((n/\tau)^2 + n|\F|/\tau)$. 
\end{enumerate}
\end{lemma}

We conclude with a brief recap of the algorithm for solving the weak isolating-cuts problem on graph $G$ and vertex subset $R$.
The algorithm first sets $\alpha=n^{-1/4}$, $w=n$ and $\tau=\Theta(n\alpha)=\Theta(n^{3/4})$.
It then constructs an $(\alpha,w)$-friendly cut sparsifier $H$ of $G$ using \Cref{lemma:friendly-cut-sparsifier},
and uses its explicit description to find (without additional cut queries),
all the minimum $s,t$-cuts (meaning for all $s,t\in R$) that are $\alpha$-friendly.
Next, to find the minimum $s,t$-cuts that are $\alpha$-unfriendly,
the algorithm applies the $(\tau,R)$-star contraction procedure.
By \Cref{lemma:tau-sigma-star-contraction}, each minimum $(s,R\setminus s)$-cut
that is also a minimum $s,t$-cut for some $t\in R$,
is preserved in the contraction $G'$ with constant probability, 
and this success probability is further amplified to $1-1/\poly(n)$ using $O(\log n)$ repetitions.
Notice that since both steps are based on contractions,
the algorithm's final output can be simply the minimum $s,t$-cut found throughout the steps and repetitions.
The query complexity of both steps is $\tO(n^{7/4})$
by \Cref{lemma:friendly-cut-sparsifier,lemma:tau-sigma-star-contraction}.

\subsection{Related Work}

\paragraph*{Gomory-Hu Tree construction} 
Our algorithm for all-pairs minimum cut is based on constructing a Gomory-Hu tree \cite{GH61}.
It has long been known that a Gomory-Hu tree can be constructed using $n-1$ calls to minimum $s,t$-cut procedure \cite{GH61}.
In recent years, there has been significant progress on this problem, culminating in algorithms that make $O(\polylog n)$ calls to minimum $s,t$-cut procedures and require $\tilde{O}(m)$ time elsewhere \cite{AKT20,JPS21,AKT21,AKT22a,AKT22b,AKLPST22,ALPS23,AKLP+25}.
Over roughly the same time span, there has been significant progress on faster algorithms (in running time, not cut queries) for the minimum $s,t$-cut problem itself,
and currently the best algorithm runs in $m^{1+o(1)}$ time \cite{CKLPGS25}.

\paragraph*{Cut Query Algorithms}
The study of cut-query algorithms was initiated in \cite{RSW18},
motivated by the relation between cut queries and submodular minimization.
and by now a few results are known for this model.
For the global minimum cut problem,
$O(n)$ randomized cut queries suffice in unweighted graphs \cite{RSW18,AEGLMN22}, matching the lower bound \cite{GPRW20, LLSZ21},
and $\tO(n)$ randomized cut queries suffice in weighted graphs \cite{MN20}.
For deterministic algorithms, $\tO(n^{5/3})$ cut queries suffice in unweighted graphs \cite{ASW25}.
For the easier problem of determining whether a graph is connected, it is known that $O(n)$ randomized queries suffice \cite{CL23,LC24}.

Yet another aspect of efficiency in this model is the round complexity,
which refers to the number of parallel rounds of queries required to solve a problem. 
For example, for every $r\ge 1$, a global minimum cut can be found
using $\tO(n^{1+1/r})$ cut queries in $O(r)$ rounds \cite{KK25}.
Additionally, there exists a one round randomized algorithm for connectivity that uses $\tO(n)$ queries \cite{ACK21}.

Meanwhile, for the minimum $s,t$-cut problem,
the known algorithms use $\tO(n^{5/3})$ cut queries \cite{RSW18,ASW25},
leaving a substantial gap to the lower bound of $\Omega(n)$ queries.
Another recent work has examined the problem of approximating the maximum cut using cut queries \cite{PRW24}.

\subsection{Future Work}

\paragraph*{Gap Between All-Pairs Minimum Cut and Minimum $s,t$-Cut}
Existing algorithms for the all-pairs minimum cut problem achieve time complexity that is comparable, namely, within a polylogarithmic factor,
of computing a single minimum $s,t$-cut \cite{AKT20,JPS21,AKT22a,AKT22b,AKLPST22,ALPS23}. 
Our work leaves a substantial gap between these two problems in the cut-query model.

A promising direction for narrowing this gap is developing efficient algorithms for the weighted minimum $s,t$-cut problem in the cut-query model. 
An efficient algorithm for this problem would immediately yield an efficient solution to the isolating-cuts problem, and hence to the all-pairs minimum cut problem using existing techniques.
Interestingly, there is a round-complexity lower bound for solving the minimum $s,t$-cut problem in weighted graphs,
as an algorithm that runs in $r$ rounds requires $\Omega(n^2/r^5)$ cut queries \cite{ACK19}.

\paragraph*{Applicability to Other Computational Models}
One surprising aspect of cut-query algorithms is that they often translate to efficient algorithms in other computational models.
For example, the global minimum cut algorithms of \cite{RSW18,MN20} also yield efficient algorithms for finding a minimum cut in the streaming and sequential models.
One reason for the wide applicability of cut-query algorithms is that they often use the following three-step approach:
(1) apply structural analysis to contract the graph while preserving some desired property;
(2) recover all the edges of the contracted graph; and
(3) solve the problem on the smaller graph.
This approach is often useful in other computational models
since a contracted graph offers a succinct representation that facilitates faster computation.
Our techniques for the all-pairs minimum cut employ this strategy
and construct two succinct data structures,
namely, a friendly cut sparsifier and a $(\tau,\F)$-star contraction,
that together encode all minimum $s,t$-cuts using only $\tO(n^{7/4})$ edges.
Hence, algorithms that construct these data structures efficiently in other computational models, such as dynamic, streaming, or parallel models,
may yield new algorithms for the minimum $s,t$-cut problem in these settings.

\section{Preliminaries}
\label{sec:preliminaries}
\paragraph{Cut Query Primitives}
Throughout the proof we use the following cut query primitives to find the weight of edges between two sets of vertices, $w(E(S,T))\coloneqq \sum_{e\in E}w_e$.
Notice that when the graph is unweighted then $w(E(S,T))= |E(S,T)|$.
\begin{claim}
  \label{claim:s-t-num-edges}
  Let $S,T \subseteq V$ be two disjoint sets.
  It is possible to find $w(E(S,T))$ using $O(1)$ non-adaptive cut queries.
\end{claim}
\begin{proof}
  Observe that $\mintcut_G(S\cup T) = \mintcut_G(S)+\mintcut_G(T)-2w(E(S,T))$.
  Hence, we can find $w(E(S,T))$ using three queries, one for each of $\mintcut_G(S)$, $\mintcut_G(T)$, and $\mintcut_G(S\cup T)$.  
\end{proof}
\begin{claim}
  \label{claim:one-edge-S-T}
  Given an unweighted graph $G$ on $n$ vertices, and two disjoint sets $S,T\subseteq V$, one can sample an edge uniformly from $E(S,T)$ using $O(\log n)$ cut queries.
\end{claim}
\begin{proof}
  At each stage the algorithm randomly partitions the vertex set $T$ into two sets of roughly equal size $T_1,T_2$, and then find $|E(S,T_i)|$ using $O(1)$ cut queries by \Cref{claim:s-t-num-edges}.
  It then recurses on either $T_1$ or $T_2$ with probability proportional to $|E(S,T_i)|$, i.e. $T_1$ is chosen with probability $|E(S,T_1)|/(|E(S,T_1)|+|E(S,T_2)|)$.
  Notice that this process terminates in $O(\log n)$ rounds and yields an endpoint $t\in T$ of the edge, but does not recover a corresponding endpoint $s\in S$.
  To find the other endpoint, the algorithm repeats the process, but now partitioning $S$ to find an edge in $E(S,t)$.
  It is straightforward to see that this indeed yields a uniform edge of $E(S,T)$ in $O(\log n)$ cut queries.
  This concludes the proof.
\end{proof}

\paragraph{Edge Perturbation}
Throughout the paper, we use a perturbation technique to ensure that the minimum $s,t$-cuts in the graph are unique.
We do this by adding an edge set $\tilde{E}$ of all possible edges $(u,v) \in \binom{V}{2}$, where each edge gets a random weight $w(u,v)$ sampled uniformly from $\set{1/n^{10},\ldots,n^{7}/n^{10}}$.
Since those edges are defined externally, and are not a part of the graph, we can access them freely without having to use any cut queries.
Note that this is slightly different from the perturbation technique used in previous work, where the perturbation is applied  only to edges that exist in the graph.
This is impossible to use in the cut-query model, as algorithms do not have access to the edges of the graph, and hence they cannot perturb existing edges.
However, our method adds noise more broadly, and hence only increases the probability that the minimum cuts in the perturbed graph are unique.
\begin{claim}[\cite{BENW16,AKT20,AKT22a}]
  Given an unweighted graph $G=(V,E)$ with $n$ vertices, if one adds a random perturbation sampled uniformly from $\set{1/n^{10},\ldots,n^{7}/n^{10}}$ to all edges in $G$, then with probability at least $1-n^{-3}$ every minimum $s,t$-cut in the perturbed graph $\tilde{G}=(V,E\cup \tilde{E})$ is unique.
\end{claim}
\begin{corollary}
  \label{corollary:perturbation-edges}
  Let $\tilde{E} = \binom{n}{2}$ be the set of all possible edges in $G$, each with weight sampled uniformly from $\set{1/n^{10},\ldots,n^{7}/n^{10}}$.
  Then, with probability at least $1-n^{-3}$ every minimum $s,t$-cut in the perturbed graph $\tilde{G}=(V,E\cup \tilde{E})$ is unique.
\end{corollary}
We note that in general we perform most algorithms directly on the graph $G$, and then add the perturbation edges $\tilde{E}$ only for computing the minimum $s,t$-cuts.
Throughout the paper, when we write that an algorithm is given a perturbation edge set $\tilde{E}$, we assume that the edges guarantee unique minimum $s,t$-cuts and that the algorithm fails otherwise.
Finally, note that given a contracted graph $G'=(V',E')$ of $G$ and a cut $S\subseteq V'$ we have that $\mintcut_{\tilde{G}'}(S)=\mintcut_{\tilde{G}}(S)$, where $\tilde{G},\tilde{G'}$ are the perturbed versions of $G,G'$ respectively.
Hence, if no edge of the unique minimum $s,t$-cut in $\tilde{G}$, $\C\subseteq V$, is contracted in $\tilde{G}'$ then $\C$ is also the unique minimum $s,t$-cut in $\tilde{G}'$.

\paragraph{Cut Sparsification}
Our algorithm uses two types of cut sparsifiers.
The first, the Nagamochi-Ibaraki (henceforth NI) sparsifier \cite{NI92} is used to handle small cuts.
Informally, this sparsifier is a subgraph that preserves the value of all cuts of value at most $k$.
\begin{definition}%
A \emph{$k$-NI sparsifier} of an unweighted graph $G=(V,E_G)$ with parameter $k\in[n]$ is a subgraph $H=(V,E_H)$ that satisfies
  \begin{equation*}
    \forall C\subseteq V,
    \qquad
    \min \{ \mintcut_G(C), k \} 
    = \min \{\mintcut_H(C), k \}
    .
  \end{equation*}
\end{definition}
\begin{lemma}[Lemma 2.1 of \cite{NI92}]
  Let $G=(V,E)$ be an unweighted graph with $n$ vertices, and for all $k\in[n]$ let $T_k$ be a maximal spanning forest of $G\setminus(\cup_{j<k}T_j)$.
  Then, $H=(V,\cup_{j=1}^k T_j)$ is a $k$-NI sparsifier of $G$.
\end{lemma}
It is known that a $k$-NI sparsifier can be computed in $O(m)$ time for every $k\in [n]$ \cite{NI92}.
We show how to construct such a sparsifier using $\tO(kn)$ cut queries.
\begin{claim}
  \label{claim:nagamochi-ibaraki-sparsifier}
  Given an unweighted graph $G=(V,E)$ with $n$ vertices, one can construct a $k$-NI sparsifier $H$ of $G$ using $\tO(k n)$ cut queries.
\end{claim}
\begin{proof}
  We show how to construct a maximal spanning forest $T$ of a given graph $G$ using $\tO(n)$ cut queries.
  Then, removing the edges in $T$ from $G$ and repeating the process $k$ times yields a Nagamochi-Ibaraki sparsifier that preserves all cuts of value at most $k$ using $\tO(kn)$ cut queries.

  To construct a maximal spanning forest $T$ of $G$, use the following procedure.
  Begin by initializing $T$ to have no edges.
  At each step pick some connected component $C$ of $T$, find some edge $e\in E(C,V\setminus C)$ using \Cref{claim:one-edge-S-T}, and add it to $T$.
  Notice that the number of connected components in $T$ decreases by one in each iteration, and hence the procedure terminates after at most $n-1$ iterations.
  This concludes the proof of the claim.
\end{proof}
The second type of cut sparsifier we use is defined as follows.
\begin{definition}%
A \emph{$(1\pm\epsilon)$-cut sparsifier} a graph $G$ is a weighted subgraph $H$
that satisfies
  \begin{equation*}
    \forall S\subseteq V,
    \qquad
    (1-\epsilon)\cdot\mintcut_G(S) 
    \leq \mintcut_H(S) 
    \leq (1+\epsilon)\cdot\mintcut_G(S)
    .
  \end{equation*}
\end{definition}
\begin{theorem}[\cite{RSW18}]
    \label{theorem:cut-sparsifier-construction}
    Given an unweighted graph $G=(V,E)$ with $n$ vertices and quality parameter $\epsilon\in (0,1)$, one can construct a $(1+\epsilon)$-cut sparsifier $H$ of $G$ using $\tO(\epsilon^{-2} n)$ cut queries.
    The algorithm is randomized and succeeds with probability $1-\poly(n)$.
\end{theorem}

\section{Isolating Cuts}
\label{sec:isolating-cuts}
In this section we prove a slightly extended version of \Cref{lemma:isolating-cuts}, stated as follows.
\begin{lemma}
    \label{lemma:isolating-cuts-full}
    Given cut-query access to an unweighted graph $G=(V,E)$ on $n$ vertices, a set of vertices $R\subseteq V$ with maximum degree $d=\max_{v\in R} d_G(v)$, and a set of perturbation edges $\tilde{E}$ that guarantee unique minimum $s,t$-cuts, one can solve the weak isolating-cuts problem using $\tO(\min\{nd, n^{7/4}\})$ cut queries.
    The algorithm is randomized and succeeds with probability $1-1/\poly(n)$.

    Furthermore, if the algorithm is given in addition all the edges incident to vertices of degree at most $100n^{3/4}$, an $(n^{-1/4},n)$-friendly cut sparsifier, and an $n^{3/4}$-NI sparsifier, 
    then it can solve the weak isolating-cuts problem using $\tO(n^{1/2} + |R|n^{1/4})$ cut queries.
\end{lemma}

The proof requires the following claim,
whose proof appears at the end of this section.
\begin{claim}
    \label{claim:equal-minimum-terminal-cut}
    Let $G=(V,E)$ be a perturbed unweighted graph, $R\subseteq V$ be a set of vertices, and $G'=(V,E\setminus (R\times R))$.
    For $v\in R$, denote by $S_v^G$ the minimum $(v,R\setminus v)$-cut in $G$.
    Then $S_v^G=S_v^{G'}$.
\end{claim}
\begin{proof}[Proof of \Cref{lemma:isolating-cuts-full}]
    The algorithm for finding the weak isolating cuts is as follows.
    If $d\le 100 n^{3/4}$, known by a preliminary round of $O(n)$ queries, solve the (stronger) isolating-cuts problem by constructing a $d$-NI sparsifier of $G$ using \Cref{claim:nagamochi-ibaraki-sparsifier}, perturb it with $\tilde{E}$, and then find each minimum $(v,R\setminus v)$-cut in the perturbed sparsifier.
    Otherwise, construct an $(n^{-1/4},n)$-friendly cut sparsifier $H$ of $G$ using \Cref{lemma:friendly-cut-sparsifier}, and let $\tilde{H}$ be the perturbed version of $H$ with added edges $\tilde{E}$.
    For each vertex $v\in R$, let $S_v^f\subseteq V$ be the minimum $(v,R\setminus v)$-cut in the perturbed sparsifier $\tilde{H}$.

    Then, the algorithm performs $r=O(\log n)$ repetitions of $(100n^{3/4},R)$-star contraction on $G$ to obtain graphs $\{G_i\}_{i=1}^r$. 
    Denote by $G_i'=(V,E_i')$ the graph obtained from $G_i$ by removing all edges internal to $R$.
    We now explain how to recover the edges of $G_i'$.
    Given two vertex sets $S,T\subseteq V$ one can recover a single edge from $E(S,T)$ using $O(\log n)$ cut queries by \Cref{claim:one-edge-S-T}.
    Therefore, the entire edge set of the graph of $G_i'$ can be recovered using $\tO(|E_i'|)$ cut queries, by iteratively recovering the sets $E(R,V\setminus R)$ and $E(V\setminus R, V\setminus R)$ one edge at a time.
    It then perturbs each $G_i'$ to get a graph $\tilde{G}_i'$ with unique minimum $s,t$-cuts, and sets $S_v^i \subseteq V$ to be the minimum $(v,R\setminus v)$-cut in $\tilde{G}_i'$.
    Finally, the algorithm returns $S_v \gets \arg\min_{A\in \set{S_v^f}\cup\set{S_v^i}_i} \mintcut_G(A)$.

    We now argue that the algorithm returns every minimum $(v,R\setminus v)$-cut that is also a minimum $v,u$-cut for some $u\in R$.
    Fix one such cut $C\subseteq V$.
    Begin by noting that $|\cutedges{C}{H}|,|\cutedges{C}{G_i}|\ge |\cutedges{C}{G}|$ since these graphs are contractions of $G$, therefore it suffices to argue that $C\in\set{S_v^f}\cup\set{S_v^i}_{i=1}^r$ with high probability.
    If $C$ is $n^{-1/4}$-friendly then the algorithm will find it in $H$.
    Otherwise, i.e. $C$ is $n^{-1/4}$-unfriendly, it is preserved in $G_i$ with constant probability by \Cref{lemma:tau-sigma-star-contraction} since $\tau/(100 \alpha) = n\ge  \deg(s),\deg(t)$ where the equality is from $\tau=100n^{3/4},\alpha=n^{-1/4}$.
    Repeating this $O(\log n)$ times, we find that with probability at least $1-n^{-10}$, $C$ is preserved in at least one of the graphs $G_i$.
    Using a union bound over all vertices $w\in R$, the probability that at least one of the graphs $G_i$ preserves the minimum $(w,R\setminus w)$-cut, if it is also a minimum $w,u$-cut for some $u\in R$, is at least $1-n^{-10}|R|\ge 1-n^{-9}$.
    Finally, by \Cref{claim:equal-minimum-terminal-cut} the minimum $(v,R\setminus v)$-cut in $\tilde{G}_i$ is equal to the minimum $(v,R\setminus v)$-cut in $\tilde{G}_i'$ and hence we can find it in the perturbed graph $\tilde{G}_i'$.
    Therefore, using a union bound with the probability that the construction of the friendly sparsifier succeeds, the algorithm returns the minimum $(v,R\setminus v)$-cut with probability at least $1-1/\poly(n)$.

    It remains to analyze the query complexity of the algorithm.
    If $d\le 100n^{3/4}$, then constructing the NI sparsifier uses $\tO(dn)$ cut queries by \Cref{claim:nagamochi-ibaraki-sparsifier}.
    Otherwise, recovering the $(n^{-1/4},n)$-friendly cut sparsifier $H$ takes $O((1/n^{-1/4})n \sqrt{n})=O(n^{7/4})$ cut queries by \Cref{lemma:friendly-cut-sparsifier}.
    Then, performing the star contraction takes $\tO(n)$ cut queries since each vertex samples at most one neighboring vertex.
    To bound the number of queries needed to recover the edges of $G_i'$, notice that each $G_i'$ has at most $O\Big((n^{1/2}+|R|n^{1/4})+n^{7/4} \Big)$ edges by \Cref{lemma:tau-sigma-star-contraction}.
    Therefore, recovering each $G_i'$ takes $\tO(n^{7/4})$ cut queries using \Cref{claim:one-edge-S-T}.
    Hence, the overall query complexity of the algorithm is $\tO(\min\{nd,n^{7/4}\})$.

    Finally, if the algorithm is given all edges incident to vertices of degree at most $n^{3/4}$, an $(n^{-1/4},n)$-friendly cut sparsifier, and an $n^{3/4}$-NI sparsifier then it only needs to recover the edges incident to $\Star$.
    By \Cref{lemma:tau-sigma-star-contraction}, there are $\tO((n^2/\tau^2+|R|(n/\tau)))=\tO(n^{1/2} + |R|n^{1/4})$ such edges, where we used $\tau=\Theta(n^{3/4})$.
    Since recovering a single edge takes $O(\log n)$ cut queries using \Cref{claim:one-edge-S-T}, the overall query complexity of the algorithm is $\tO(n^{1/2} + |R|n^{1/4})$.
\end{proof}

It remains to prove \Cref{claim:equal-minimum-terminal-cut}.
\begin{proof}[Proof of \Cref{claim:equal-minimum-terminal-cut}]
    Let $C$ be a $(v,R\setminus v)$-cut in $G$.
    The edge set of the cut can be decomposed as
    \begin{equation*}
        E(C,V\setminus C)
        = E(v,R)\cup E(v,V\setminus (R\cup C))\cup E(C\setminus\{v\},V\setminus C)
        .
    \end{equation*}
    The second and third sets, $E(v,V\setminus (R\cup C))$ and $E(C\setminus\{v\},V\setminus C)$ are the same in $G$ and in $G'$.
    Furthermore, $E(v,R)$ does not depend on $C$ and hence for every two $(v,R\setminus v)$-cuts $C_1,C_2$ we have $\mintcut_G(C_1)>\mintcut_G(C_2)$ if and only if $\mintcut_{G'}(C_1)>\mintcut_{G'}(C_2)$ since $E(v,R)$ contributes the same additive quantity to both sides.
    Therefore, the minimizer over all $(v,R\setminus v)$-cuts in $G$ is the same as the minimizer over all $(v,R\setminus v)$-cuts in $G'$ (and vice versa).
\end{proof}

\section{\texorpdfstring{$(\tau,\F)$-Star Contraction}{(tau,F)-Star Contraction}}
\label{sec:sigma-tau-star-contraction}
In this section we prove \Cref{lemma:tau-sigma-star-contraction},
which provides the guarantees of the $(\tau,\F)$-star contraction procedure.
We will need the following claim, whose proof is similar to that of \cite[Proposition 4.5]{AEGLMN22} and appears in \Cref{sec:expectation-c-r-ratio}.

\begin{claim}
    \label{claim:expectation-c-r-ratio}
    Let $G=(V,E)$ be a simple $n$-vertex graph, $\C\subseteq E$ be some non-trivial cut of $G$, and $\Uc\eqdef \set{v\in V\mid \text{ $v$ is $\alpha$-unfriendly}}$.
    Then, for any $v\in V$,
    \begin{equation*}
        \Exp{\cross{v}{\Star}/d_{\Star}(v) \mid \Uc \cap \Star = \emptyset}{} \le \frac{\crossbar{v}{-\Uc}}{\deg(v)-|\Uc|}
        ,
    \end{equation*}
    where $\crossbar{v}{-\Uc}$ is the number of edges incident to $v$ that are in $\cutedges{\C}{}$ but not incident to any vertex in $\Uc$.
\end{claim}
\begin{proof}[Proof of \Cref{lemma:tau-sigma-star-contraction}]
    Throughout the proof set $p\coloneqq 800\cdot \log n/\tau$ and fix some $\alpha$-unfriendly minimum $s,t$-cut $\C\subseteq V$ of $G$ such that $s,t\in\F$ and $\min\set{\deg(s),\deg(t)}\ge \tau/(100\alpha)$.
    We begin by showing the correctness guarantee of the procedure, i.e. that no edges in $\cutedges{\C}{}$ is contracted with probability $2/5$.
    Denote a set $\Star$ as good, if it satisfies $\max_{v\in H} \cross{v}{\Star}/d_{\Star}(v)\le 1/2$ and $\sum_{v\in H} \cross{v}{\Star}/d_{\Star}(v) \le 1/2$.
    If a set $\Star$ is good, then,
    \begin{align*}
        \Probability{\text{the cut $\C$ is preserved}}
        &\ge 
        \Probability{\text{the cut $\C$ is preserved} \mid \text{$\Star$ is good}} \cdot \Probability{\text{$\Star$ is good}}
        \\
        &\ge \Big(1-\frac{1}{2}\Big)^{\lceil (1/2) /(1/2)\rceil} \cdot \Probability{\text{$\Star$ is good}}
        =  \frac{1}{2} \cdot \Probability{\text{$\Star$ is good}},
    \end{align*}
    where the first inequality is from the law of total probability, and the second is from \Cref{proposition:probability-optimization}.
    Therefore, we now focus on showing that $\Star$ is good with high constant probability.

    We start by bounding the probability that $\sum_{v\in H} \cross{v}{\Star}/d_{\Star}(v) > 1/2$.
    Observe that since $\C$ is a minimum $s,t$-cut, at most one vertex (namely $s$ or $t$) is $\alpha$-unfriendly.
    Henceforth, assume without loss of generality that $s$ is $\alpha$-unfriendly.
    The probability that $s$ is sampled into $\Star$ is $p= \Theta(\log n/\tau) = \Theta(1/\log n)$.
    For the rest of the proof, assume that $\set{s} \not \subseteq \Star$, and that the algorithm fails otherwise.
    Notice that under this assumption no edge incident to $s$ is contracted since the procedure only contracts edges in $H\times \Star$, and $s\not\in \Star, H$ as $H$ does not contain any vertex in $F$.
    By \Cref{claim:expectation-c-r-ratio} we have,
    \begin{align} 
        \Exp{\sum_{v\in H} \frac{\cross{v}{\Star}}{d_{\Star}(v)} \Big| \set{s} \cap \Star = \emptyset}{\Star}
        &= \sum_{v\in H} \Exp{\frac{\cross{v}{\Star}}{d_{\Star}(v)} \Big| \set{s} \cap \Star = \emptyset }{\Star}
        &  \le\sum_{v\in H} \frac{\crossbar{v}{-s}}{\deg(v)-\indic{(s,v)\in E}}
        \label{eq:expectation-sum}
        .
    \end{align}
    Therefore, we need to bound $\sum_{v\in H} \crossbar{v}{-s}$.
    Notice that $\sum_{v\in H} \crossbar{v}{-s}$ is at most twice the number of edges in $\cutedges{\C}{}$ that are not incident to $s$.
    Observe that $|E(\C,V\setminus \C)| \le \min\{\deg(s),\deg(t)\}$ as  $\C$ is a minimum $s,t$-cut and hence,
    \begin{align*}
        \sum_{v\in H} \crossbar{v}{-s} 
        &=2(|E(\C,V\setminus \C)| - |E(\set{s},V\setminus \C)|)
        = 2(|E(\C,V\setminus \C)| - \cross{s}{})
        \\
        &\le 2(\deg(s) - (1-\alpha)\deg(s))
        \le 2\alpha \deg(s) 
        \le \tau/50
        ,
    \end{align*}
    where the first inequality is from the friendliness ratio of $s$ and the last is from $\deg(s)\le \tau/(100\alpha)$ by the theorem statement.
    Furthermore, $\deg(v)\ge \tau$ for every $v\in H$, and therefore
    \begin{equation*}
        \sum_{v\in H} \frac{\crossbar{v}{-s}}{\deg(v)-\indic{(s,v)\in E}}
        \le \frac{\tau/50}{\tau-1}
        \le \frac{1}{20}
        ,
    \end{equation*}
    for every large enough $n$ (and hence $\tau$).
    Plugging back into \Cref{eq:expectation-sum} we have that,
    \begin{align*}
        \Exp{\sum_{v\in H} \frac{\cross{v}{\Star}}{d_{\Star}(v)} \Big| \set{s} \cap \Star = \emptyset}{\Star}
        \le\sum_{v\in H} \frac{\crossbar{v}{-s}}{\deg(v)-\indic{(s,v)\in E}}
        \le  \frac{1}{20}
        ,
    \end{align*}
    and by Markov's inequality we have,
    \begin{equation*}
        \ProbOn{\sum_{v\in H} \frac{\cross{v}{\Star}}{d_{\Star}(v)} \ge \frac{1}{2}}{\Star}
        \le \frac{1}{10}
        .
    \end{equation*}
    Using a union bound on the event that $\set{s}\not \subseteq \Star$ we have that with probability at least $4/5$,  $\sum_{v\in H} \cross{v}{\Star}/d_{\Star}(v) \le 1/2$.
    Notice that when this event occurs then $\max_{v\in H} \cross{v}{\Star}/d_{\Star}(v)\le 1/2$ as well since every term in the sum is non-negative.
    Therefore, the set $\Star$ is good with probability at least $4/5$.
    This concludes the proof of the correctness guarantee of the procedure, i.e., that no edge in  $E(\C,V\setminus \C)$ is contracted with probability at least $2/5$.

    To conclude the proof we provide the size guarantee.
    Partition the vertices of $G'$ into $\Star,\F$ and $L=V\setminus (\Star\cup \F)$.
    Recall that the algorithm samples each vertex into $\Star$ independently with probability $p$.
    By the Chernoff bound (\Cref{theorem:chernoff}), the probability that $d_{\Star}(v)\ge 0.9 \deg(v)\cdot p$ is at least,
    \begin{equation*}
        \Probability{d_{\Star}(v) \ge 720 \deg(v)\log n /\tau} 
        \ge 1-\exp(-(1/10)^2 \deg(v)800\log n /(2\tau)) 
        \ge 1-1/n^4
        ,
    \end{equation*}
    where the last inequality is from $\deg(v)\ge \tau$ as $v\in H$.
    Therefore, $d_{\Star}(v)\ge 0.9 \cdot p \deg(v)$ for all $v\in H$ with probability at least $1-1/n^3$.
    Therefore, every vertex in $H$ is contracted to some $p\in \Star$ with high probability.
    Hence, every vertex $v\in L$ that was not contracted has degree at most $\tau$.
    For the rest of the proof we assume that this event holds.
    
    The total number of edges in the graph is bounded by the number of edges internal to $\Star\cup\F$ plus the total number of edges incident to $L$.
    Since, $|\Star|\le \tO(n/\tau)$ with high probability the number of edges internal to $\Star\cup \F$ is at most $\tO\big((n/\tau +|\F|)^2\big)$.
    Furthermore, since the degree of the vertices of $L$ is bounded by $\tau$, the total number of edges incident on $L$ is at most $O(n\tau)$.
    Hence, the total number of edges in $G'$ is at most $\tO\Big( (n/\tau+|\F|)^2+n\tau \Big)$ with high probability.
    Discounting the edges in $\F\times\F$ we have that the number of edges in $G'$ is at most $\tO\Big( (n/\tau)^2 + |\F|n/\tau + n\tau \Big)$ with high probability following the same argument.
    Similarly, discounting both the edges in $\F\times\F$ and the edges incident to low degree vertices, we have that the number of edges in $G'$ is at most $\tO\Big( (n/\tau)^2 + |\F|n/\tau \Big)$ with high probability.
\end{proof}

\subsection{Proof of Claim \ref{claim:expectation-c-r-ratio}}
\label{sec:expectation-c-r-ratio}

We now prove \Cref{claim:expectation-c-r-ratio},
using similar arguments to \cite[Proposition 4.5]{AEGLMN22}.

\begin{proof}
    Begin by noting that by convention when $d_{\Star}(v)=0$ then the expectation is trivially $0$ and the claim trivially holds.
    For the rest of the proof condition on $d_{\Star}(v)>0$.
    Denote by  $X_1,\ldots,X_{\cross{v}{}},Y_1,\ldots,Y_{\deg(v)-\cross{v}{}}$ the random variables indicating if the edges incident to $v$ are sampled into $\Star$.
    The variables $X_i$ correspond to the edges incident to $v$ that are in $\C$, and the variables $Y_i$ correspond to the edges incident to $v$ that are not in $\C$.

    Let $X = \sum_{i} X_i$ and $Y = X + \sum_{i} Y_i$.
    Notice that the theorem statement is equivalent to showing that,
    \begin{align*}
        \Exp{X/Y \mid Y >0, \Uc \cap \Star = \emptyset}{} 
        \le  \frac{\crossbar{v}{-\Uc}}{\deg(v)-|\Uc|} 
        .
    \end{align*}
    Denote the number of edges incident to $v$ that are also incident to $\Uc$ by $a$.
    Now fix $Y=b$ for some $0<b\le \deg(v)-a$ (since we condition on no vertex of $\Uc$ being sampled into $\Star$) and notice that,
    \begin{align*}
        \Exp{X/Y \mid Y=b, \Uc \cap \Star = \emptyset}{} 
        & = \frac{1}{b} \cdot \Exp{X \mid Y=b, \Uc \cap \Star = \emptyset}{}
        .
    \end{align*}
    Notice that the expectation is upper bounded by the case when all the edges of $v$ that are incident to $\Uc$ are not part of the cut $\C$, i.e. they all correspond to variables in $\set{Y_i}_i$ and not in $\set{X_i}_i$.
    In this case, that $X$ is a hypergeometric random variable corresponding to making $b$ draws from $\deg(v)-a$ edges, of which $\crossbar{v}{-\Uc}$ are marked.
    Therefore, $\Exp{X \mid Y=b, \Uc \cap \Star = \emptyset}{} \le b\cdot \frac{\crossbar{v}{-\Uc}}{\deg(v)-a}$.
    Hence, for every $b$ we have that, $\Exp{X/Y \mid Y=b, \Uc \cap \Star = \emptyset}{} \le \frac{\crossbar{v}{-\Uc}}{\deg(v)-|\Uc|}$, where we used that $a\le |\Uc|$.

    Since this is true for all $b \ge 0$, we find that $\Exp{\cross{v}{\Star}/d_{\Star}(v) \mid \Uc \cap \Star = \emptyset}{} \le \frac{\crossbar{v}{-\Uc}}{\deg(v)-|\Uc|}$.
    This concludes the proof.
\end{proof}

\section{\texorpdfstring{Single-Source Minimum Cut} {Single-Source Minimum Cuts}}
\label{sec:single-source-cuts}
In this section we prove \Cref{lemma:single-source-min-cut},
which is copied here for convenience.

\singleSourceMinCut*

The lemma is based on \Cref{algorithm:single-source-min-cut}.
An important subroutine of the algorithm is \Cref{algorithm:unfriendly-min-cut}, which handles the case of $0.4$-unfriendly minimum $p_i,v$-cuts.
Throughout this section we denote the value of the minimum $s,t$-cut in the graph $G$ by $\lambda_{s,t}(G)$.
The following claim, whose proof is provided at the end of this section, details the guarantees of the unfriendly minimum cut algorithm.
\begin{claim}
    \label{claim:unfriendly-min-cut}
    Given an unweighted graph $G=(V,E)$, a partition of the vertices into sets $V_1,\ldots,V_k$, and a set of pivots $\{p_i\}_{i=1}^k$ such that $p_i\in V_i$, for every $i\in[k], v\in V_i$ \Cref{algorithm:unfriendly-min-cut} returns the minimum $p_i,v$-cut if it is $0.4$-unfriendly, and returns some larger cut otherwise.
    The algorithm uses $\tO(n^{1.75})$ cut queries and succeeds with high probability.
\end{claim}

\begin{algorithm}
    \caption{Unfriendly Minimum $s,t$-Cuts}
    \label{algorithm:unfriendly-min-cut}
    \begin{algorithmic}[1]
        \State \textbf{Input:} An unweighted graph $G=(V,E)$, vertex partition $V_1,\ldots,V_k$, a set of pivots $\{p_i\}_{i=1}^k$, and perturbation edges $\tilde{E}$.
        \Procedure{Single-Source-Unfriendly-Cuts}{$G,\{V_i\},\{p_i\},\tilde{E}$}
            \State $\epsilon,\delta\gets 1/100, r\gets O(\log n)$
            \State $G'\gets $ a quality $(1\pm \epsilon)$-cut sparsifier of $G$ using \Cref{theorem:cut-sparsifier-construction}
            \For{$i\in [k]$}
                \State $\tilde{c}(v)\gets 2\lambda_{p_i,v}(G')/(1-\epsilon)$ for all $v\in V_i$
            \EndFor
            \For{$i\in [k]$}
                \For{$j\in[r]$}
                    \State $T_j \gets \{v\in V_i \mid \tilde{c}(v)\ge (1+\delta)^j\}\cup \set{p_i}$ 
                    \State $\{S_v^j\}_{v\in T_j} \gets $ Isolating-Minimum-Cuts($G_i, T_j,\tilde{E}$)
                \EndFor
                \State $S_v \gets \arg\min_{S\in \{S_v^1,\ldots,S_v^r,V\setminus S_{p_i}^1,\ldots, S_v^r,V\setminus S_{p_i}^r\}} \mintcut_{\tilde{G}}(S)$ for all $v\in V_i$
            \EndFor
            \State \Return $\{S_v\}_{v\in V}$.
        \EndProcedure
    \end{algorithmic}
\end{algorithm}

\begin{algorithm}
    \caption{Single-Source Minimum $s,t$-Cuts}
    \label{algorithm:single-source-min-cut}
    \begin{algorithmic}[1]
        \State \textbf{Input:} An unweighted graph $G=(V,E)$, vertex partition $V_1,\ldots,V_k$, a set of pivots $\{p_i\}_{i=1}^k$, and perturbation edges $\tilde{E}$.
        \Procedure{Single-Source-Min-Cuts}{$G,\{V_i\},\{p_i\},\tilde{E}$}
            \State $H \gets $ a friendly $(2/5, n)$-cut sparsifier of $G$
            \State $\tilde{H} \gets$ perturb $H$ with edges $\tilde{E}$ such that the minimum $s,t$-cuts are unique
            \State $\{S^f_v\}_{v\in V} \gets$ compute the minimum $p_i,v$-cuts in $\tilde{H}$ for all $i\in [k]$ and $v\in V_i$
            \State $\{S^{u}_v\}_{v\in V} \gets$ \textsc{Single-Source-Unfriendly-Cuts}($G,\{V_i\},\{p_i\},\tilde{E}$)
            \State $\tilde{G} \gets$ perturb $G$ with edges $\tilde{E}$ such that the minimum $s,t$-cuts are unique
            \For{$v\in V$}
                \State $S_v \gets \arg\min_{S\in \{S^f_v,S^{u}_v\}} \mintcut_{\tilde{G}}(S)$
            \EndFor
            \State \Return $\{S_v\}_{v\in V}$ \Comment{return the cut value as well}
        \EndProcedure
    \end{algorithmic}
\end{algorithm}

\begin{proof}[Proof of \Cref{lemma:single-source-min-cut}]
    We begin by proving the correctness of \Cref{algorithm:single-source-min-cut}.
    Fix some vertex $v\in V_i$ and a pivot $p_i\in V_i$.
    Notice that by \Cref{lemma:friendly-cut-sparsifier}, we have $\lambda_{p_i,v}(H)\ge \lambda_{p_i,v}(G)$ and similarly by the guarantees of \Cref{claim:unfriendly-min-cut}, we have $\mintcut_G(S^{u}_v)\ge \lambda_{p_i,v}(G)$.
    Furthermore, the minimum $p_i,v$-cut $S_v$ is either $0.4$-friendly, in which case it is found in $H$ by \Cref{lemma:friendly-cut-sparsifier}, or it is $0.4$-unfriendly, in which case it is found by \Cref{claim:unfriendly-min-cut}.
    Hence, the minimum between the cuts $S^f_v$ and $S^{u}_v$ yields the unique minimum $p_i,v$-cut.

    To analyze the query complexity of the algorithm, observe that constructing the $(0.4,n)$-friendly cut sparsifier $H$ requires $\tO(n^{3/2})$ cut queries by \Cref{lemma:friendly-cut-sparsifier}.
    Then, perturbing $H$ with edges $\tilde{E}$, and finding the minimum cuts requires no cut queries.
    By \Cref{claim:unfriendly-min-cut}, the unfriendly minimum cuts procedure uses $\tO(n^{1.75})$ cut queries.
    Therefore, the overall query complexity of the algorithm is $\tO(n^{1.75})$.
    Finally, the algorithm succeeds with high probability since both \Cref{claim:unfriendly-min-cut} and \Cref{lemma:friendly-cut-sparsifier} succeed with high probability, and the perturbation edges $\tilde{E}$ are chosen such that the minimum cuts in the perturbed graph are unique with high probability.
\end{proof}

\subsection{Proof of Claim \ref{claim:unfriendly-min-cut}}
The proof of the claim is similar to that of Theorem 1.4 in \cite{AKT22b}, and requires the following lemmas.
\begin{lemma}
    \label{lemma:unfriendly-v}
    If the minimum $p,v$-cut, $S\subseteq V$ such that $v\in S$, is $0.4$-unfriendly, and $v$ is $0.4$-unfriendly, then for all $v'\in V\setminus \{v\}$, $\lambda_{p,v'}(G)\le 0.8\cdot \lambda_{p,v}(G)$.
\end{lemma}
\begin{lemma}
    \label{lemma:unfriendly-p}
    If the minimum $p,v$-cut, $S\subseteq V$ such that $v\in S$, is $0.4$-unfriendly, and $p$ is $0.4$-unfriendly, then for all $v'\in (V\setminus S) \setminus\{p\}$, $\lambda_{p,v'}(G)\le 0.8\cdot \lambda_{p,v}(G)$.
\end{lemma}
\begin{proof}[Proof of \Cref{claim:unfriendly-min-cut}]
    Throughout the proof, assume that the sparsifier construction succeeds and that the algorithm fails otherwise.
    Fix some vertex $v\in V_i \setminus \set{p}$ and a pivot $p_i\in V_i$.
    Begin by noting that the minimum $p_i,v$-cut returned by the algorithm is at least the value of the minimum $p_i,v$-cut in $G$ by \Cref{lemma:isolating-cuts-full}, even if the minimum $p_i,v$-cut is friendly.
    We show that if the minimum $p_i,v$-cut is $0.4$-unfriendly, then the algorithm return it with high probability.
    Since by \Cref{lemma:isolating-cuts-full} the weak isolating-cuts procedure returns a cut that is at least the value of the minimum $p_i,v$, it suffices to show that \Cref{algorithm:unfriendly-min-cut} finds the minimum $p_i,v$-cut in at least one of its iterations.

    Denote the unique minimum $p_i,v$-cut by $S_v$.
    Notice that if $T_j \cap (V_i\setminus S_v) = \set{p_i}$ for some iteration $j$ then the minimum $v,T_j\setminus v$ is the minimum $p_i,v$-cut, and similarly if $T_j \cap S_v = \set{v}$.
    Therefore, it remains to show that one of these events occurs in some iteration $j$.
    Let $j^*$ be the index such that $(1+\delta)^{j^*} \le \lambda_{p_i,v}(G') < (1+\delta)^{j^*+1}$.
    We now split the proof into two cases, depending on whether the $p_i$ or $v$ is $0.4$-unfriendly.
    Notice that at most one of them can be $0.4$-unfriendly since we consider a minimum $p_i,v$-cut.

    \textbf{Case 1:} $p_i$ is $0.4$-unfriendly.
    In this case, by \Cref{lemma:unfriendly-p}, for all $u\in (V_i\setminus S_v) \setminus\{p_i\}$, $\lambda_{p_i,u}(G)\le 0.8\cdot \lambda_{p_i,v}(G)$.
    Therefore, $\tilde{c}(u)\le (1+3\epsilon)\lambda_{p_i,u}(G)\le 0.8\cdot (1+3\epsilon)\lambda_{p_i,v}(G)\le \lambda_{p_i,v}/(1+\delta) \le (1+\delta)^j$.
    Therefore, $u\not\in T_j$ and we obtain the desired event that $T_j \cap (V_i\setminus S_v) = \set{p_i}$.

    \textbf{Case 2:} $v$ is $0.4$-unfriendly.
    In this case, by \Cref{lemma:unfriendly-v}, for all $u\in S_v \setminus\{v\}$, $\lambda_{p_i,u}(G)\le 0.8\cdot \lambda_{p_i,v}(G)$.
    Therefore, $\tilde{c}(u)\le (1+3\epsilon)\lambda_{p_i,u}(G)\le 0.8\cdot (1+3\epsilon)\lambda_{p_i,v}(G)\le \lambda_{p_i,v}/(1+\delta) \le (1+\delta)^j$.
    Therefore, $u\not\in T_j$ and we obtain the desired event that $T_j \cap S_v = \set{v}$.

    Hence, the algorithm succeeds in finding the minimum $p_i,v$-cut as long as all the weak isolating-cuts instances succeed.
    Notice that the algorithm solves $k\cdot r = \tO(n)$ instances of the weak isolating-cuts problem, and since each instance succeeds with high probability, the overall success probability of the algorithm is high.

    To conclude the proof we analyze the query complexity of the algorithm.
    The first step of the algorithm is to construct a quality $(1\pm \epsilon)$-cut sparsifier $G'$ of $G$ using \Cref{theorem:cut-sparsifier-construction}, which requires $\tO(n)$ cut queries.
    Then, the algorithm solves $k\cdot r$ instances of the weak isolating-cuts problem.
    To do so efficiently, it recovers the following data from the graph once, and then passes it to all the instances.
    First, it recovers all the edges incident to vertices of degree at most $n^{3/4}$ using $\tO(n^{7/4})$ cut queries using \Cref{claim:one-edge-S-T}.
    The algorithm then constructs a $(n^{-1/4},n)$-friendly cut sparsifier $H$ of $G$ using \Cref{lemma:friendly-cut-sparsifier}, which requires $\tO(n^{7/4})$ cut queries.
    Finally, it constructs an $n^{7/4}$-NI sparsifier of $G$ using \Cref{claim:nagamochi-ibaraki-sparsifier}, which requires $\tO(n^{7/4})$ cut queries.
    Hence, the query complexity of each weak isolating-cuts instance is at most $\tO(n^{1/2} + |T_j|n^{1/4})\le \tO(n^{1/2} + |V_i|n^{1/4})$.
    Since $\sum_i |V_i|=n$ and $k\le n$, the total query complexity of all weak isolating-cut instances is at most $\tO(n^{3/2})$.
    Therefore, the overall query complexity of the algorithm is $\tO(n^{1.75})$.
\end{proof}

\section{Friendly Cut Sparsifiers}
\label{sec:friendly-cut-sparsifiers}
In this section we prove \Cref{lemma:friendly-cut-sparsifier}, showing how to construct a friendly cut sparsifier using $\tO(\alpha^{-1}n\sqrt{w})$ cut queries.
We begin by describing the expander decomposition procedure, which is the main technical tool used for constructing the friendly cut sparsifier.
\subsection{Expander Decomposition}
We begin by introducing the notion of expander decomposition and several useful facts about it.
Let $G=(V,E)$ be some graph, and assume we are given a demand vector $\dd\in \R^{|V|}_+$.
Then, $G$ is a said to be a \emph{$(\phi,\dd)$-expander} if,
\begin{equation*}
    \forall S\subseteq V,
    \qquad
    \Phi_G^{\vec{d}}(S) 
    \eqdef
    \frac{\mintcut_G(S)}{\min(\dd(S),\dd(V\setminus S))} 
    \ge \phi
    .
\end{equation*}
We use the following lemma for finding a decomposition of a graph into a collection of $(\phi,\dd)$-expanders, it is based on a reduction to the $\BalCutPrune$ procedure which has been used for constructing expander decomposition in the deterministic sequential setting \cite{CGLNPS20,LS21}.
The proof of the lemma is provided in \Cref{sec:expander-decomposition}.
\begin{restatable}{lemma}{expanderDecomposition}
    \label{lemma:expander-decomposition}
    Let $G=(V,E,w)$ be a weighted graph with $1\le w_e \le U$ for all $e\in E$, $\dd\in \R^{|V|}_+$ be a demand vector such that $\dd(v)\in\set{0}\cup[1,U]$, and $\phi\in (0,1)$ a parameter.
    There exists an algorithm that partitions $V$ into sets $V_1,\ldots,V_k$ such that, 
    \begin{enumerate}
        \item For every $i\in [k]$, $G[V_i]$ is a $(\phi,\dd_{V_i})$-expander, where $\dd_{V_i}$ is the vector $\dd$ limited to the vertices in $V_i$.
        \item The total weight of inter-cluster edges is at most $O(\phi \dd(V)\log (mU))$.
    \end{enumerate}
    The algorithm uses $\tO(n)$ randomized cut queries and succeeds with probability $1-1/\poly(n)$.
\end{restatable}

\subsection{Friendly Cut Sparsifier Construction}
In this section we provide a construction of a friendly cut sparsifier.
Recall that a cut $S\subseteq V$ is $\alpha$-friendly if for every vertex $v\in V$ we have $\cross{v}{}/d(v)\le 1-\alpha$.
Note that originally, friendly cut sparsifiers were defined with a fixed $\alpha=0.4$ \cite{AKT22b}.
We show how to implement the algorithm of \cite{AKT22b} using $O(\alpha^{-1}n\sqrt{w})$ cut queries.
The proof requires the following lemma, whose proof follows similar lines to \cite{AKT22b} and is deferred to \Cref{sec:friendly-cut-sparsifier-akt}.
\begin{algorithm}
    \caption{Friendly Cut Sparsifier Construction}
    \label{algorithm:friendly-cut-sparsifier}
    \begin{algorithmic}[1]
        \State \textbf{Input:} An unweighted graph $G=(V,E)$
        \Procedure{Friendly-Cut-Sparsifier}{$G,\alpha$}
            \State $V_1,\ldots,V_k \gets (\phi,\dd)$-expander decomposition of $G$ with $\phi=0.01$ and $\dd(v)=\phi^{-1}\sqrt{w}$ for all $v\in V$
            \For{$i\in [k]$}
                \State $R_i \gets \Big\{ v\in V_i \mid d_G(v)<\max \{ 10\alpha^{-1}\sqrt{w}, 4\alpha\cdot |E_{G}(\set{v}, V\setminus H_i)|\} \Big\}$
                \State $H_i' \gets H_i \setminus R_i$
            \EndFor
        \State $G' \gets$ contract in $G$ each $H_i'$ into a single vertex $v_i$
        \State recover the edges of $G'$ using $O(|E'|)$ cut queries iteratively using \Cref{claim:one-edge-S-T}.
        \State \Return $G'$.
        \EndProcedure
    \end{algorithmic}
\end{algorithm}
\begin{restatable}{lemma}{friendlySparsifierAKT}
    \label{lemma:friendly-sparsifier-AKT}
    Given an unweighted graph $G$ on $n$ vertices, applying \Cref{algorithm:friendly-cut-sparsifier} on $G$ with parameter $\alpha$ and $w$, returns a friendly $(\alpha,w)$-cut sparsifier $G'$ of $G$ with $\tO(\alpha^{-1}n\sqrt{w})$ edges.
\end{restatable}
\begin{proof}[Proof of \Cref{lemma:friendly-cut-sparsifier}]
    Begin by constructing a $(\phi,\dd)$-expander decomposition of $G$ by applying \Cref{lemma:expander-decomposition} with $\phi=0.01$ and $\dd(v)=\phi^{-1}\sqrt{w}$ for all $v\in V$.
    Notice that this requires $\tO(n)$ cut queries.
    Then, for each expander $H_i$ in the decomposition, remove all vertices $v\in H_i$ such that $d_{G}(v) <  \max \{ 10\alpha^{-1}\sqrt{w}, 4\alpha|E_{G}(\set{v}, V\setminus H_i)|\}$.
    This requires at most $O(n)$ cut queries since one can find $d_G(v),|E_G(\set{v}, T)|$ for every $v\in V$ and $T\subseteq V$ using $O(1)$ cut queries by \Cref{claim:s-t-num-edges}.
    Therefore, computing the shaved expanders $H_i'$ requires at most $\tO(n)$ cut queries.
    To recover the graph $G'=(V',E')$, the algorithms recovers its entire edge set using $\tO(|E'|)$ cut queries by iteratively finding edges using \Cref{claim:one-edge-S-T}.
    Since the number of edges in $G'$ is at most $\tO(\alpha^{-1}n\sqrt{w})$ by \Cref{lemma:friendly-sparsifier-AKT}, we find that the total number of cut queries used by the algorithm is $\tO(\alpha^{-1}n\sqrt{w})$.
    Finally, the success probability of the algorithm is determined solely by the expander decomposition, which succeeds with probability $1-n^{-10}$ by \Cref{lemma:expander-decomposition}.
    This concludes the proof of \Cref{lemma:friendly-cut-sparsifier}.
\end{proof}

{\small
  \bibliographystyle{alphaurl}
  \bibliography{cut-gomory-hu}
} %

\appendix
\section{Concentration Inequalities}
\label{sec:concentration-inequalities}
\begin{theorem}\label[theorem]{theorem:chernoff}
  Let $X_1,\ldots,X_m\in [0,a]$ be independent random variables.
    For every $\delta \in [0,1]$ and $\mu \geq \mathbb{E}\left[\sum_{i=1}^{m}X_i\right]$, we have
        \begin{align}
            \nonumber
            \mathbb{P}\left[
                \left|
                    \sum_{i=1}^{m}X_i - \mathbb{E}\left[\sum_{i=1}^{m}X_i\right]
                \right|
                \geq \delta \mu
            \right]
            \leq
            2\exp
            \left(
                -\frac{\delta^2\mu}{3a}
            \right)
            .
        \end{align}
\end{theorem}

\section{Friendly Cut Sparsifier Proof}
\label{sec:friendly-cut-sparsifier-akt}
In this section, we prove \Cref{lemma:friendly-sparsifier-AKT}.
We begin by restating the theorem for convenience.
\friendlySparsifierAKT*

The proof is based on the following two claims.
Notice that combining these two claims immediately yields \Cref{lemma:friendly-sparsifier-AKT}.
\begin{claim}[Equivalent of Claim 2.3 in \cite{AKT22b}]
    \label{claim:friendly-sparsifier-correctness}
    Let $S\subseteq V$ be an $\alpha$-friendly cut of $G$ with at most $w$ edges.
    Then, the output of \Cref{algorithm:friendly-cut-sparsifier} preserves $S$ (i.e. no edge of $E(S,V\setminus S)$ is contracted).
\end{claim}

\begin{claim}[Equivalent of Claim 2.4 in \cite{AKT22b}]
    \label{claim:friendly-sparsifier-size}
    The number of edges in the output of \Cref{algorithm:friendly-cut-sparsifier} is at most $\tO(\alpha^{-1}n\sqrt{w_j})$.
\end{claim}

\begin{proof}[Proof of \Cref{claim:friendly-sparsifier-correctness}]
  The proof follows the same lines as the proof of Claim 2.1 in \cite{AKT22b}.
  Notice that an edge $(x,y)\in E(S,V\setminus S)$ is contracted if and only if both endpoints $x,y$ are contained in the same shaved expander $H_i'$.
  Let $L=H_i \cap S$ and $R=H_i \cap (V\setminus S)$ be the induced the cut on the original expander, before shaving.
  Assume without loss of generality that $x\in L, y\in R$.
  Since $H_i$ is a $(\phi,\dd)$-expander, we have,
  \begin{equation*}
    \frac{|E(L,R)|}{\min\set{\dd(L),\dd(R)}} 
    \ge \phi
    .
  \end{equation*}
  Hence, $\min\set{\dd(L),\dd(R)} \le \phi^{-1} w$ since $|E(L,R)|\le \mintcut_G(S)\le w$ by the theorem statement.
  We now lower bound  $\dd(L)$ ($\dd(R)$ is bounded by a symmetric argument).
  Notice that since the cut $S$ is friendly in $G$, we have $|E(\set{x},S)| = \deg(x) - \cross{v}{} \ge \alpha \cdot\deg(x)$ since $\cross{x}{} \le (1-\alpha)\cdot\deg(x)$.
  Furthermore, since $x$ was not shaved then $|E(\set{x},V\setminus H_i)| \le \alpha \cdot\deg(x)/4$ and,
  \begin{equation}
    \label{eq:friendly-sparsifier-lower-bound}
    |E(\set{x},L)| 
    = |E(\set{x},S\cap H_i)|
    \ge |E(\set{x},S)| - |E(\set{x},V\setminus H_i)|
    \ge (\alpha-\alpha/4)\deg(x)
    \ge 5/2 \deg(x) \sqrt{w}
    ,
  \end{equation}
  where the last inequality follows as $\deg(x) \ge 10\alpha^{-1}\sqrt{w}$ since $x$ was not shaved.
  Notice that since $G$ is a simple graph, we have that $|E(\set{x},L)| \le |L|$.
  By the definition of $\dd$, we have $\dd(L) = |L|\phi^{-1}\sqrt{w} \ge (5/2)\phi^{-1} w$, where the inequality is from \Cref{eq:friendly-sparsifier-lower-bound}.
  The argument is symmetric for $R$, hence we have $\min\set{\dd(L),\dd(R)} \ge (5/2)\phi^{-1} w$ which leads to a contradiction.
  Therefore, no edge of $E(S,V\setminus S)$ is contracted, and the algorithm preserves the cut $S$.
\end{proof}
\begin{proof}[Proof of \Cref{claim:friendly-sparsifier-size}]
  The proof follows the same lines as the proof of Claim 2.2 in \cite{AKT22b}.
  Notice that $G'$ has three types of edges,
  \begin{enumerate}
    \item The outer edges of the expander decomposition, of which there are at most $O(\phi\dd(V)\log m)=\tO(n\sqrt{w})$ by \Cref{lemma:expander-decomposition} and the value of the demand vector.
    \item The edges adjacent to vertices shaved because of their degree, i.e. those with $\deg(v) < 10\alpha^{-1}\sqrt{w}$, of which there are at most $O(\alpha^{-1}n\sqrt{w})$.
    \item Edges that are incident to vertices that were shaved because at least $\alpha^{-1} \deg(v)/4$ of their edges went outside their expander. 
    For every $v\in V$ let $d_{out}(v)$ be the cardinality of the edge set $E(\set{v},V\setminus H)$, where $H$ is the cluster to which $v$ belongs. 
    Furthermore, let $X\subseteq V$ be the set of the aforementioned shaved vertices, observe that the total number of inter-cluster edges is at most $\tO(n\sqrt{w}) = \sum_{v\in V}d_{out}(v)$ as explained above. 
    Since for every $x\in X$ we have $d_{out}(x) \ge \alpha \deg(x)/4$, we find that the number of edges incident to $X$ is at most $\tO(\alpha^{-1}n\sqrt{w})$.
  \end{enumerate}
  Summing up the three types of edges, we find that the total number of edges in $G'$ is at most $\tO(\alpha^{-1}n\sqrt{w})$.
\end{proof}

\section{Expander Decomposition}
\label{sec:expander-decomposition}
In this section we prove \Cref{lemma:expander-decomposition}, which we restate here for convenience.
\expanderDecomposition*
We use the $\BalCutPrune$ approach to find an expander decomposition of $G$.
\begin{definition}
    \label{def:bal-cut-prune}
    Given an undirected graph $G=(V,E)$, demands vector $\dd\in\R^{|V|}_+$, sparsity parameter $\phi\in(0,1]$, and approximation factor $\alpha$, the goal of the $\alpha$-approximate $\BalCutPrune$ problem is to find a partition $(A,B)$ of $V$ such that $w(E(A,B))\le \alpha\phi \min\set{\dd(A),\dd(B)}$ and, either:
    \begin{enumerate}
        \item \textbf{Cut:} $\dd(A),\dd(B) \ge \dd(V)/3$; or
        \item \textbf{Prune:} $\dd(A)\ge \dd(V)/2$ and $\Phi_{G[A]}\ge \phi$,
    \end{enumerate}
    where $\Phi_{G[A]}$ is the expansion of the induced subgraph $G[A]$ defined as,
    \begin{equation*}
        \Phi_{G[A]} 
        = \min_{S\subseteq A, 0<|S|<|A|} \frac{w(E(S,A\setminus S))}{\min{\dd(S),\dd(A\setminus S)}}
    \end{equation*}
\end{definition}

The main technical lemma needed for the proof of \Cref{lemma:expander-decomposition} is the following lemma, which shows how to solve the $\BalCutPrune$ problem using few cut queries.
\begin{lemma}
    \label{lemma:bal-cut-prune}
    Given an undirected graph $G=(V,E)$ on $n$ vertices, a subset of vertices $W\subseteq V$, demand vector $\dd\in\R^{|W|}_+$, and sparsity parameter $\phi\in(0,1]$, one can solve the $2$-approximate $\BalCutPrune$ on $G[W]$ using $\tO(|W|)$ randomized cut queries.
    The algorithm succeeds with probability $1-n^{-10}$.
\end{lemma}
\begin{proof}
    Notice that one can simulate any cut query $S\subseteq W$ in $G[W]$ by $O(1)$ cut queries in $G$, since $\mintcut_{G[W]}(S) = |E(S,W\setminus S)|$ which can be recovered using $O(1)$ cut queries in $G$ by \Cref{claim:s-t-num-edges}.
    Begin by constructing a quality $(1\pm 1/2)$-cut sparsifier $H$ of $G$, this requires $\tO(n)$ cut queries by \Cref{theorem:cut-sparsifier-construction}.
    Notice that this allows us to approximate the value $w(E(A,B))$ for any partition $(A,B)$ of $V$ up to factor $2$.
    Hence, the algorithm can find all partitions $(A,B)$ of $V$ such that $w(E(A,B))\le 2\cdot\phi \min\set{\dd(A),\dd(B)}$ and then return the first partition that satisfies either the \textbf{Cut} or \textbf{Prune} condition.
    Finally, notice that the only probabilistic part of the algorithm is the construction of the cut sparsifier, which succeeds with probability $1-n^{-10}$ by \Cref{theorem:cut-sparsifier-construction}.
    This concludes the proof of \Cref{lemma:bal-cut-prune}.
\end{proof}
\begin{proof}[Proof of \Cref{lemma:expander-decomposition}]
    The proof follows the line of the proof of Corollary 2.5 in \cite{LS21}.
    Throughout the proof we maintain a collection $\HH$ of disjoint subgraphs of $G$ which are called clusters, the collection is partitioned into \emph{active} clusters $\HH^A$ and \emph{inactive} clusters $\HH^I$.
    We also maintain a set of edges $E'$ which are outside all the clusters.

    The algorithm for the decomposition is as follows.
    While $\HH^A\neq \emptyset$, run the $2$-approximate $\BalCutPrune$ on $G[H]$ with demands $\dd_H$ and sparsity parameter $\phi$ for all $H\in \HH^A$ in parallel.
    Fix some $H\in\HH$ and let $(A,B)$ be the partition returned by the $\BalCutPrune$ procedure.
    If $\dd(A),\dd(B)\ge \dd(H)/3$ then replace $H$ with $A$ and $B$ in $\HH^A$.
    Otherwise, $\dd(A)\ge \dd(H)/2$ and $\Phi_{G[A]}\ge \phi$, hence the algorithm removes $H$ from $\HH^A$ and adds $A$ to $\HH^I$ and $B$ to $\HH^A$.
    Finally, it adds the edges $E(A,B)$ to $E'$.

    Notice that in every iteration, $\dd_H(H)$ is reduced by at least a constant factor, therefore the algorithm terminates after at most $O(\log mU)$ iterations.
    Notice that when the algorithm terminates, we are guaranteed that $\Phi_{G[H]}\ge \phi$ for every $H\in \HH^I$.
    This satisfies the first condition of the theorem.

    We now bound the total weight of inter-cluster edges.
    Notice that since the cuts found are $\phi$-sparse, in every cluster $H$ we have that $w(E_{G[H]}(A,B)) \le 2\phi \min\set{\dd_H(A),\dd_H(B)}\le 2\phi\dd_H(H)$, and summing over all partitions we find that the total weight of inter-cluster edges is at most, $2\phi \dd(V)$.
    Therefore, in $O(\log (mU))$ iterations the algorithm adds edges of total weight at most $O(\phi \dd(V)\log (mU))$ to $E'$.

    To conclude the proof, we analyze the query complexity of the algorithm.
    Notice that the query complexity of the $2$-approximate $\BalCutPrune$ on a cluster $H$ is $\tO(|H|)$ by \Cref{lemma:bal-cut-prune}.
    Since the total number of vertices in all clusters is at most $n$, the total query complexity of the algorithm in a single iteration is $\tO(n)$ and the overall query complexity is $\tO(n\log (mU))$.
    Finally, the success probability of the algorithm is determined solely by the $\BalCutPrune$ procedure, and using a union bound over all iterations we find that the algorithm succeeds with probability $1-n^{-10}\cdot \tO(n\log (mU))=1-1/\poly(n)$.
\end{proof}

\section{Gomory-Hu Algorithm}
\label{sec:gomory-hu-algorithm}
In this section we provide a reduction from a single-source minimum cut procedure algorithm to computing the Gomory-Hu tree of a graph $G$, proving \Cref{lemma:gomory-hu-tree-reduction}.
We begin by restating the lemma for convenience.
\ghtreereduction*
The algorithm is based on the general framework of \cite{AKT22b,AKT22a}.
The proof requires the following definitions.
\begin{definition}[Partial $k$-Tree]
    A $k$-partial tree of a graph $G=(V,E)$ is a partition tree of $G$, such that for every two sets $X,Y\subseteq V$ in the partition,
    \begin{itemize}
        \item For every $x_1,x_2\in X$, the minimum $x_1,x_2$-cut is larger than  $k$, and
        \item For every $x\in X,y\in Y$ the minimum $x,y$-cut in $G$ is equal to the minimum $X,Y$-cut in $T$.
    \end{itemize} 
\end{definition}
\begin{definition}[Auxiliary Graph]
    Given a partition tree of $T$ of a graph $G=(V,E)$ and a super-vertex $W\in T$, the auxiliary graph $G_W$ is the graph obtained by contracting every connected component of $T\setminus W$ into a single vertex.
\end{definition}

We begin by presenting the algorithm of \cite{AKT22a} and then explain how to implement it in the cut-query model.
\paragraph{Gomory-Hu Algorithm}
\begin{enumerate}
    \item Compute a partial $k$-tree $T$ of $G$ with $k=\sqrt{n}$. This is the initial Gomory-Hu tree that is refined throughout the algorithm.
    \item For each super vertex $V_i$, do the following:
    \begin{enumerate}
        \item If $|V_i|=1$, then stop and continue to the next super vertex. \label{lst:step:gh-stop-if-singleton}
        \item Get an auxiliary graph $G_i$ using the current tree $T$, then compute a perturbed version $\tilde{G}_i$ of $G_i$ such that the minimum cuts are unique.
        \item Pick a pivot $p\in V_i$ uniformly at random, and compute single-source minimum cuts $S_v$ from $p$ to all other vertices in $V_i$ on the perturbed graph $\tilde{G}_i$. \label{lst:step:gh-random-pivot}
        \item Denote each vertex such that $|S_v\cap V_i|\le |V_i|/2$ as good, and the rest as bad.
        \item If there are less than $|V_i|/4$ good vertices return to step \ref{lst:step:gh-random-pivot} and pick a new pivot (If this fails $20\log n$ times then abort the algorithm).
        Otherwise, continue to the next step. \label{lst:step:gh-bad-pivot}
        \item Refine $T$ based on the cuts as follows. 
        For each good vertex $u\in V_i$, assign it to the largest good cut $S_v$ such that $u\in S_v$.
        Then, let $S_1,\ldots, S_r$ be the cuts to which at least one good vertex was assigned.
        Refine the tree by taking a Gomory-Hu step vertex $V_i$ using each of the cuts above.
        This results in replacing  $V_i$ with a set of super vertices $\{S_1,\ldots, S_r,V_i\setminus (\cup_i S_i)\}$.
        \item Recurse on all the new super vertices by going to step \ref{lst:step:gh-stop-if-singleton}.
    \end{enumerate}
    \item Return the final tree $T$.
\end{enumerate}
The following result states that the algorithm above computes a Gomory-Hu tree of $G$.
\begin{lemma}[\cite{AKT22a}]
    \label{lemma:akt-gh-tree-algorithm-guarantee}
    The algorithm above returns a Gomory-Hu tree of $G$.
    Furthermore, the probability of failure in step \ref{lst:step:gh-bad-pivot} is at most $n^{-10}$.
\end{lemma}

We now prove \Cref{lemma:gomory-hu-tree-reduction}.
\begin{proof}[Proof of \Cref{lemma:gomory-hu-tree-reduction}]
    We explain how to implement the algorithm above in the cut-query model.
    The first step is to compute a partial $k$-tree of $G$ with $k=\sqrt{n}$.
    We do this by recovering a $2\sqrt{n}$-NI sparsifier of $G$ using \Cref{claim:nagamochi-ibaraki-sparsifier} and then constructing a partial $k$-tree of the sparsifier using some algorithm.
    Notice that since the sparsifier captures every cut of size at most $\sqrt{n}$, it suffices for constructing a partial $k$-tree.
    Note that this takes $\tO(n^{1.5})$ cut queries by \Cref{claim:nagamochi-ibaraki-sparsifier}.

    To proceed, instead of recursively breaking each super vertex $V_i$ until it becomes a singleton, run the algorithm inside all super vertices of a given level in parallel.
    To perform step \ref{lst:step:gh-random-pivot}, sample uniformly a pivot $p_i\in V_i$ for all $i$ in parallel.
    Then, call the single-source minimum cut algorithm in the theorem statement.
    Note that the single-source minimum cut algorithm solves the problem of computing the minimum $p_i,v$-cut for every $v\in V_i$ on $G$ directly, and not on the auxiliary graphs $\set{G_i}_i$.
    However, by the structure of the Gomory-Hu tree we are guaranteed that those cuts are the same as the minimum $p_i,v$-cuts in the auxiliary graphs $G_i$.

    For every partition $V_i$ which passes the condition in step \ref{lst:step:gh-bad-pivot}, the algorithm stops and refines the tree.
    Notice that we do not need any additional cut queries for checking the condition or the refinement step, since the cuts $S_v$ and their values were already computed in the previous step.
    
    Notice that the algorithm terminates with depth $O(\log n)$ and runs for at most $O(\log n)$ times in each level, therefore we call the single-source minimum cut algorithm at most $O(\log^2 n)$ times.
    Hence, the total number of cut queries used by the algorithm is at most $\tO(n^{1.5}+ q(n)\log^2 n)$, where $q(n)$ is the query complexity of the single-source minimum cut algorithm in the theorem statement.
    The success probability follows from \Cref{lemma:akt-gh-tree-algorithm-guarantee} and the fact that the single-source minimum cut algorithm succeeds with probability $\rho_F$.
    This concludes the proof of \Cref{lemma:gomory-hu-tree-reduction}.
\end{proof}

\end{document}